\documentclass[submission,copyright,creativecommons]{eptcs}
\providecommand{\event}{} 
\usepackage{breakurl}             
\usepackage{amsmath,amssymb,amsfonts,amsthm}
\usepackage{fancyvrb}
\usepackage{bigfoot}
\usepackage{tikz}

\newtheorem{theorem}{Theorem}
\newtheorem{corollary}{Corollary}
\usetikzlibrary{calc,trees,positioning,arrows,chains,shapes.geometric,
    decorations.pathreplacing,decorations.pathmorphing,shapes,
    matrix,shapes.symbols}

\tikzset{ basic/.style = {draw, text width=2cm, font=\sffamily,
    rectangle}, level 2/.style = {basic, rounded corners=6pt, thin,
    align=center, fill=green!60, text width=8em}, level 3/.style =
  {basic, rounded corners=6pt, thin, align=center, fill=orange!60,
    text width=8em}, level 4/.style =
  {basic, rounded corners=6pt, thin, align=center, fill=yellow!60,
    text width=12em}, level 5/.style =
  {basic, rounded corners=6pt, thin, align=center, fill=yellow!60,
    text width=8em} }

\title{Fourier Series Formalization in ACL2(r)}
\author{Cuong K. Chau
\institute{Department of Computer Science\\
The University of Texas at Austin\\
Austin, TX, USA}
\email{ckcuong@cs.utexas.edu}
\and
Matt Kaufmann
\institute{Department of Computer Science\\
The University of Texas at Austin\\
Austin, TX, USA}
\email{kaufmann@cs.utexas.edu}
\and
Warren A. Hunt, Jr.
\institute{Department of Computer Science\\
The University of Texas at Austin\\
Austin, TX, USA}
\email{hunt@cs.utexas.edu}
}
\def\titlerunning{Fourier Series Formalization in ACL2(r)}
\def\authorrunning{C. Chau, M. Kaufmann \& W. Hunt}
\begin{document}
\maketitle

\begin{abstract}
We formalize some basic properties of Fourier
series in the logic of ACL2(r), which is a variant of ACL2 that
supports reasoning about the real and complex numbers by way of
non-standard analysis. More specifically, we extend a framework for
formally evaluating definite integrals of real-valued, continuous
functions using the Second Fundamental Theorem of Calculus.
Our extended framework is also applied to functions
containing free arguments. Using this framework, we are able to prove
the orthogonality relationships between trigonometric functions, which
are the essential properties in Fourier series analysis. The sum rule
for definite integrals of indexed sums is also formalized by applying
the extended framework along with the First Fundamental Theorem of
Calculus and the sum rule for differentiation. The Fourier
coefficient formulas of periodic functions are then formalized from
the orthogonality relations and the sum rule for integration.
Consequently, the uniqueness of Fourier sums is a straightforward
corollary.

We also present our formalization of the sum rule for definite
integrals of infinite series in ACL2(r). Part of this task is to prove
the Dini Uniform Convergence Theorem and the continuity of a limit
function under certain conditions. A key technique in our proofs of
these theorems is to apply the {\it overspill principle} from
non-standard analysis.


\end{abstract}

\section{Introduction}
\label{sec:intro}

In this paper, we present our efforts in formalizing some basic
properties of Fourier series in the logic of ACL2(r), which is a
variant of ACL2 that supports reasoning about the real and complex
numbers via non-standard analysis \cite{gamboa:1999, gamboa:2001}. In
particular, we describe our formalization of the Fourier coefficient
formulas for periodic functions and the sum rule for definite
integrals of infinite series. The formalization of Fourier series will
enable interactive theorem provers to reason about systems modeled by
Fourier series, with applications to a wide variety of problems in
mathematics, physics, electrical engineering, signal processing, and
image processing.

We do not claim to be developing new mathematics.  However, as far as
we know the mechanized formalizations and proofs presented in this
paper are new.  The research contributions of this paper are twofold:
a demonstration that a mechanized proof assistant, in particular
ACL2(r), can be used to verify properties of Fourier series; and
infrastructure to support that activity, which we expect to be
reusable for future ACL2(r) verifications of continuous
mathematics. Our formalizations presented in the paper assume that
there exists a Fourier series, i.e., a (possibly infinite) sum of
sines and cosines for any periodic function. Future work could include
proving convergence of the Fourier series for any suitable periodic
function.

The proofs of {\it Fourier coefficient formulas} depend on the {\it
orthogonality relationships between trigonometric functions} and the
{\it sum rule for integration of indexed sums}.  A key tool for
proving these properties is the Second Fundamental Theorem of Calculus
(FTC-2).  Cowles and Gamboa \cite{cowles:2014} implemented a framework
for formally evaluating definite integrals of real-valued continuous
functions using FTC-2.  However, their framework is restricted to
unary functions, while formalizing Fourier coefficient formulas
requires integration for indexed families of functions $f_n(x)$, which
we represent as $f(x, n)$.  We call such $n$ a {\it free argument}.
Hence, we extend the FTC-2 framework of Cowles and Gamboa to apply to
functions with free arguments. We call the extended framework the {\it
FTC-2 evaluation procedure}. One may expect the usual ACL2 {\it
functional instantiation} mechanism to apply, by using pseudo-lambda
expressions~\cite{acl2-doc:lemma-instance} to handle the free
arguments.  However, in ACL2(r) there are some technical issues and
restrictions on the presence of free arguments in functional
substitutions, which make functional instantiation not
trivial \cite{gamboa:2007}. We describe these issues in detail and
show how we deal with them in Section \ref{sec:procedure}. Once the
FTC-2 evaluation procedure is built, we can use it to prove the
orthogonality relationships between trigonometric functions. The sum
rule for definite integrals of indexed sums is also formalized by
applying the FTC-2 evaluation procedure along with the First
Fundamental Theorem of Calculus (FTC-1) and the sum rule for
differentiation. The Fourier coefficient formulas for periodic
functions are then verified using the orthogonality relations and the
sum rule for integration.  Consequently, the {\it uniqueness of
Fourier sums} is a straightforward corollary of the Fourier
coefficient formulas.

The other main contribution of our work is the formalization of the
{\it sum rule for definite integrals of infinite series} under two
different conditions. This problem deals with the {\it convergence}
notion of a sequence of functions. We consider two types of
convergence: {\it pointwise convergence} and {\it uniform
convergence}. Our formalization requires that a sequence of partial
sums of real-valued continuous functions converges uniformly to a {\it
continuous limit function} on the interval of interest. We approach
this requirement in two ways, corresponding to two different
conditions. One way is to prove that if a sequence of continuous
functions converges pointwise on a closed and bounded interval, then
it converges uniformly on that interval, given that the sequence is
monotonic and the limit function is continuous. This is known as the
{\it Dini Uniform Convergence Theorem} \cite{luxemburg:1971}. Another
way is to prove that if a sequence of continuous functions is not
required to be monotonic but converges uniformly to some limit
function on the interval of interest, then the limit function is also
continuous on that interval. A key technique in our proofs for both
cases is to apply the {\it overspill principle} from non-standard
analysis \cite{goldblatt:1998, keisler:1976}. Thus, we also formalize
the overspill principle in ACL2(r) and apply this principle to prove
Dini's theorem and the continuity of the limit function as mentioned.




\begin{center}

\tikzstyle{l} = [draw, -latex']

\begin{tikzpicture}[
  level 1/.style={sibling distance=40mm},
  edge from parent/.style={->,draw},
  >=latex]

\node [level 2] (n0) {FTC-2};
\node [level 2, right of = n0, node distance = 11em] (n1) {FTC-1};
\node [level 2, below of = n0, node distance = 4em] (n2)
{Orthogonality relations};
\node [level 2, below of = n1, node distance = 4em] (n3)
{Sum rule for integration};

\path [l] (n0) -- (n2);
\path [l] (n0) -- (n3);
\path [l] (n1) -- (n3);

\node[level 4, below right of = n2, node distance = 8em] (n4)
{Fourier coefficient formulas};

\path [l] (n2) -- (n4);
\path [l] (n3) -- (n4);

\node[level 2, below of = n4, node distance = 4em] (n5)
{Uniqueness of Fourier sums};

\path [l] (n4) -- (n5);

\node[level 3, right of = n1, node distance = 11em] (n6)
{Overspill principle};
\node[level 5, right of = n4, node distance = 16.5em] (n8)
{Sum rule for integration of infinite series};

\path [l] (n3) -- (n8);
\path [l] (n6) -- (n8);

\end{tikzpicture}

\end{center}

\noindent
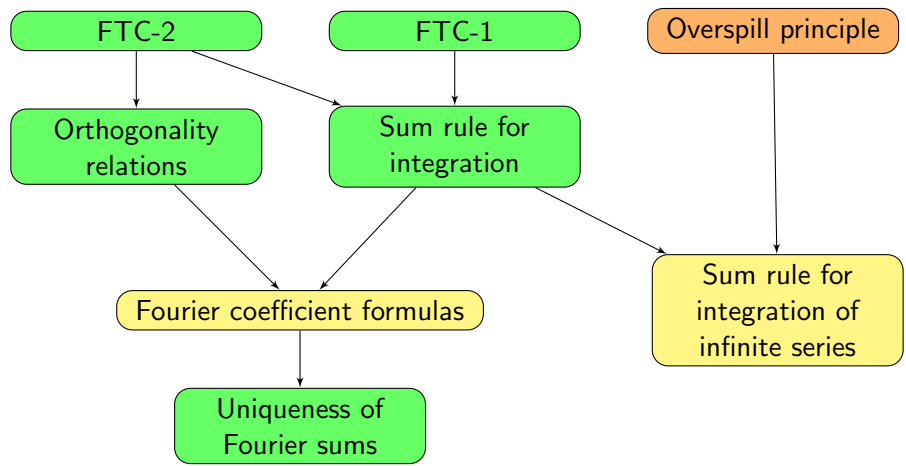
Figure 1. Overview of the formalization of the Fourier coefficient
formulas and the sum rule for definite integrals of infinite series.

\medskip

Figure 1 gives an overview of the work presented in this
paper. The remainder of the paper is organized as follows.
Section \ref{sec:acl2r} reviews some basic notions of non-standard
analysis in ACL2(r) that we use later in the paper.
Section \ref{sec:ftc} reviews two versions of the Fundamental Theorem
of Calculus that we extend to support our Fourier series
formalization. Section \ref{sec:procedure} describes the FTC-2
evaluation procedure as an extended framework for applying FTC-2 to
functions with free arguments. The formalization of the orthogonality
relations for trigonometric functions and the sum rule for definite
integrals of indexed sums are described in
Sections \ref{sec:orthog-relations}
and \ref{sec:sum-rule-for-integration} respectively. The formalization
of the Fourier coefficient formulas and the uniqueness of Fourier sums
are described in Section \ref{sec:fourier-coefficients}. The preceding
results apply to finite series, but as we look ahead to dealing with
infinite Fourier series, we take a step in Section \ref{sec:int-inf},
which presents our formalization of the sum rule for definite
integrals of infinite series. Finally, Section \ref{sec:conclusions}
concludes the paper and points out some possible future work.

\section{Basic Non-Standard Analysis Notions in ACL2(r)}
\label{sec:acl2r}
Here we review basic notions of non-standard analysis in ACL2(r)
that are used in the remainder of this paper. All notions introduced here are
considered {\it non-classical}, while functions whose definitions do
not depend on any of these notions are {\it classical}.
Let $x$ be a real number.

\begin{itemize}

\item A primitive notion is that $x$ is {\it standard}, which
  intuitively means that $x$ is a ``traditional'' real number.  In
  particular, $x$ is standard if it can be defined. For example, 1,
  -2, 3.65, $\pi$, $e^5$, and $\sqrt{2}$ are standard. A natural
  number is considered standard if it is finite, otherwise it is
  non-standard. We will refer to the standard notion of natural
  numbers when stating the overspill principle in
  Section~\ref{sec:int-inf}. We feel free to {\it relativize} our
  quantifiers. For example, ``$\forall^{st} n \ldots$'' means ``for
  all standard $n \ldots$'', and ``$\exists^{\neg st} n \ldots$''
  means ``there exists non-standard $n \ldots$''.

\item $x$ is {\tt i-small} ({\it infinitesimal}) iff
  $|x| < r$ for all positive standard reals $r$.

\item $x$ is {\tt i-large} iff $|x| > r$ for all positive
  standard reals $r$.

\item $x$ is {\tt i-limited} ({\it finite}) iff $|x| < r$ for some
  positive standard real r.

\item $x$ is {\tt i-close} ($\approx$) to a real $y$ iff $(x - y)$ is
  {\tt i-small}.

\item Suppose $x$ is {\tt i-limited}.  Then {\tt standard-part}($x$), or simply
  st($x$), is the {\it unique standard real} that is {\tt i-close} to
  $x$.

\end{itemize}

\section{Fundamental Theorem of Calculus}
\label{sec:ftc}
This section reviews two versions of the Fundamental Theorem of
Calculus that we need to extend to functions with
free arguments, as part of our Fourier series formalization. The
two versions are sometimes called the First and Second Fundamental
Theorem of Calculus.

{\bf First Fundamental Theorem of Calculus (FTC-1):}
Let $f$ be a real-valued continuous function on the interval $[a,
b]$. We can then define
a corresponding function $g(x)$ as follows: $g(x) = \int_a^x f(t) dt$.
Then $g'(x) = f(x)$ for all $x \in [a, b]$.

{\bf Second Fundamental Theorem of Calculus (FTC-2):}
If $f$ is a real-valued continuous function on $[a, b]$ and $g$ is an
antiderivative of $f$ on $[a, b]$, i.e., $g'(x) = f(x)$ for all $x \in
[a, b]$, then
\[
\int_{a}^b f(x)dx = g(b) - g(a).
\]

In the next two sections, we extend FTC-2 to
functions with free arguments and apply it to prove the orthogonality
relations of trigonometric functions, respectively. The extension of
FTC-1 and its application to the sum rule for definite
integrals of indexed sums is described in Section
\ref{sec:sum-rule-for-integration}.

\section{FTC-2 Evaluation Procedure}
\label{sec:procedure}
This section describes how we apply the FTC-2 theorem to evaluate
definite integrals of real-valued continuous functions $f$ in terms of
their antiderivatives $g$, even when $f$ and $g$ contain {\it free
arguments}, that is, arguments other than the variable with respect to
which we perform integration or differentiation. In particular, we
extend the existing FTC-2 framework \cite{cowles:2014} to
functions with free arguments, and call the extended
framework the {\it FTC-2 evaluation procedure}. This procedure
consists of the following steps:

\begin{itemize}

\item Prove that $f$ returns real values on $[a, b]$.

\item Prove that $f$ is continuous on $[a, b]$.

\item Specify a real-valued antiderivative $g$ of $f$ and prove that
  $f$ is the derivative of $g$ on $[a, b]$; i.e., prove that $g$
  returns real values and $g'(x) = f(x)$ for all $x \in$ $[a, b]$.

\item Formalize the integral of $f$ on $[a, b]$ as the Riemann integral.

\item Evaluate the integral of $f$ on $[a, b]$ in terms of
  $g$ by applying the FTC-2 theorem.

\end{itemize}

The first two steps are trivial in comparison to the last three.  In
the following subsections, we describe the challenges manifest in the
last three steps and how we tackle them.

\subsection{Automatic Differentiator}

In order to apply the FTC-2 evaluation procedure to evaluate the
definite integral of a function $f$, we need to specify and prove the
correctness of a real-valued antiderivative $g$ of $f$. The specifying
task can be done by appealing to a computer algebra system such as
{\it Mathematica} \cite{mathematica:2015}. Notably, we must
mechanically check in ACL2(r) that $f$ is indeed the derivative of
$g$. Fortunately, we don't have to prove this manually for every
function. An {\it automatic differentiator} (AD) implemented by Reid
and Gamboa \cite{reid:itp-2011, reid:acl2-2011} symbolically computes
the derivative $f$ of the input function $g$ and automatically derives
a proof demonstrating the correctness of the differentiation, i.e.,
automatically proves the following formula:
\[
f(x) \approx \frac{g(x) - g(y)}{x - y},
\]

\noindent for all $x$ and $y$ in the domain of $g$ such that $x$ is standard,
$x \approx y$, but $x \neq y$. For example, the user can employ the AD
to prove that $f(x) = n \cos(nx)$ is the derivative of $g(x)
= \sin(nx)$ with respect to $x$, by calling the
macro {\tt defderivative} with the input function $g$ as follows:
\begin{Verbatim}
(defderivative sine-derivative
  (acl2-sine (* n x)))
\end{Verbatim}

\noindent The following theorem is then introduced and proved automatically:
\begin{Verbatim}
(defthm sine-derivative
  (implies (and (acl2-numberp x)
                (acl2-numberp (* n x))
                (acl2-numberp y) 
                (acl2-numberp (* n y))
                (standardp x)
                (standardp n) (acl2-numberp n)
                (i-close x y) (not (equal x y)))
           (i-close (/ (- (acl2-sine (* n x))
                          (acl2-sine (* n y)))
                       (- x y))
                    (* (acl2-cosine (* n x))
                       (+ (* n 1) (* x 0))))))
\end{Verbatim}

The AD requires using the symbol $x$ as the name of the variable with
respect to which the (partial) derivative is computed.
Notice that the hypotheses {\tt (acl2-numberp (* n x))} and
{\tt (acl2-numberp (* n y))} in the above theorem are redundant
since they can be implied
from the set of hypotheses {\tt (acl2-numberp x)},
{\tt (acl2-numberp y)} and {\tt (acl2-numberp n)}. In addition,
the above theorem states that the derivative of $\sin(nx)$ is
$\cos(nx)(n*1 + x*0)$, which indeed equals $n \cos(nx)$. This AD
does not perform such simplifications. Nevertheless, the user can
easily prove the desired theorem from the one generated by the
macro {\tt defderivative}.

\subsection{Formalizing the Riemann Integral with Free Arguments}
\label{subsec:riemann-int}

We formalize the definite integral of a function as the Riemann
integral, following the same method as implemented by
Kaufmann \cite{kaufmann:2000}. When functions contain free arguments,
this formalization encounters a problem with functional instantiations
of non-classical theorems containing these functions. We will describe
the problem in detail and how we deal with it. Let's consider the
following definition of the Riemann integral of a {\it unary}
function, which uses an ACL2(r) utility, {\tt
defun-std}~\cite{gamboa:2007}, for introducing classical functions
defined in terms of non-classical functions. Note that {\tt defun-std}
defines a function which is only guaranteed to satisfy its definition
on standard inputs.
\begin{Verbatim}
(defun-std strict-int-f (a b)
  (if (and (inside-interval-p a (f-domain))
           (inside-interval-p b (f-domain))
           (< a b))
      (standard-part (riemann-f (make-small-partition a b)))
    0))
\end{Verbatim}

\noindent The form above introduces the Riemann integral of a
function $f$ as a classical function, even though it contains two
non-classical functions, {\tt standard-part} and {\tt
make-small-partition}\footnote{We use the non-classical function {\tt
make-small-partition} to partition a closed and bounded interval into
subintervals each of infinitesimal length.}. The
proof obligation here is to prove the integral returns standard values
with standard inputs.  More specifically, we need to prove that the
standard part of the Riemann sum of $f$, for any partition of
$[a, b]$ with standard endpoints into infinitesimal-length subintervals,
returns standard values. This is
true only if that Riemann sum is limited.  In fact, for a generic
real-valued continuous {\it unary} function {\tt
rcfn}, this limited property was proven for a corresponding Riemann
sum, as follows~\cite{kaufmann:2000}.
\begin{Verbatim}
(defthm limited-riemann-rcfn-small-partition
  (implies (and (standardp a)
                (standardp b)
                (inside-interval-p a (rcfn-domain))
                (inside-interval-p b (rcfn-domain))
                (< a b))
           (i-limited (riemann-rcfn (make-small-partition a b)))))
\end{Verbatim}

We are now interested in extending the above theorem for functions
containing free arguments using functional instantiation with
pseudo-lambda expressions.  Unfortunately, free arguments are not
allowed to occur in pseudo-lambda expressions in the functional
substitution since the theorem we are trying to instantiate is
non-classical and the functions we are trying to instantiate are
classical; the following example shows why this requirement is
necessary~\cite{gamboa:2007}. For an arbitrary classical function
$f(x)$, the following is a theorem.
\[
{\it standardp}(x) \Rightarrow {\it standardp}(f(x))
\]

\noindent Substitution of $\lambda (x).(x + y)$ for $f$ into the above
formula yields the formula
\[
{\it standardp}(x) \Rightarrow {\it standardp}(x + y)
\]

\noindent which is not valid, since the free argument {\it y} can be
non-standard.

Instead of using functional instantiation, we prove the limited
property of Riemann sums (as discussed above) from scratch by applying
the following theorem.

\begin{theorem}[The boundedness of Riemann sums \cite{kaufmann:2000}]
\label{thm:riemann-sum}
Assume that there exist finite values m and M such that
\[
m \leq f(t) \leq M, \mbox{ for all } t \in [a, b].
\]

\noindent Then the Riemann sum of f over $[a, b]$ with any partition
P = $\{x_0,x_1,\ldots,x_n\}$
is bounded by
\[
m(b - a) \leq \sum_{i=1}^n f(t_i)(x_i - x_{i-1}) \leq M(b - a)
\]

\noindent where $t_i \in [x_{i-1}, x_i]$, $x_0 = a$, and $x_n = b$.
\end{theorem}

%
%

From Theorem \ref{thm:riemann-sum}, proving the Riemann sum of $f$
over $[a, b]$ is bounded reduces to proving $f$ is bounded on that
interval.  Given a {\it specific} real-valued continuous function $f$,
it is usually straightforward to specify the bounds of $f$ on a closed
and bounded interval.  The problem becomes more challenging when
applying to {\it generic} real-valued continuous functions since it is
impossible to find either their minimum or maximum.  However, the
boundedness of these functions on a closed and bounded interval still
holds by the {\it extreme value theorem}.  But again, this was just
proven for {\it unary} functions \cite{cowles:2014}.  We also want to
apply this property to functions with free arguments. Our solution at
this point is to re-prove the extreme value theorem and consequently
the limited property of Riemann sums for generic functions with free
arguments.  Since the number of free arguments is varied, it would be
troublesome to prove the same properties independently for
each number of free arguments.  Indeed, we just need to add only one
{\it extra} argument representing a list of the free arguments to the
constrained functions and re-prove the concerned non-classical
theorems. The necessary hypotheses for the extra argument can be added
throughout the proof development. Note that
non-classical theorems proven for the new constrained functions with
only one extra argument added can also be derived for functions with
an arbitrary number of free arguments, using functional and ordinary
instantiation. (See lemmas {\tt
limited-riemann-f-small-partition-lemma} and {\tt
limited-riemann-f-small-partition} below for an example of how this
works.)  The question is how can we avoid the problem of the
appearance of free arguments in functional instantiations of
non-classical theorems as described above?  The trick is to treat the
extra argument in the constrained functions as a list of free
arguments.  Thus, no free argument appears in the functional
instantiations.  To illustrate the proposed technique, let us
investigate the constrained function {\tt rcfn-2} below.  It contains
one main argument {\tt x} and one extra argument {\tt arg}.
\begin{Verbatim}
(encapsulate
 ((rcfn-2 (x arg) t)
  (rcfn-2-domain () t))

 ;; Our witness real-valued continuous function is the
 ;; identity function of x. We ignore the extra argument arg.

 (local (defun rcfn-2 (x arg) (declare (ignore arg)) (realfix x)))
 (local (defun rcfn-2-domain () (interval nil nil)))

 ... ;; Non-local theorems about rcfn-2 and rcfn-2-domain
 )
\end{Verbatim}

We then prove the extreme value theorem for {\tt rcfn-2} and
consequently the limited property of the Riemann sum of {\tt rcfn-2},
using the same proofs for the case of {\it unary} function {\tt rcfn}
existing in the ACL2 community books \cite{acl2-books}, file {\tt
books/nonstd/integrals/continuous-function.lisp}. The limited property
of the Riemann sum of {\tt rcfn-2} is stated as follows:
\begin{Verbatim}
(defthm limited-riemann-rcfn-2-small-partition
  (implies (and (standardp arg)
                (standardp a)
                (standardp b)
                (inside-interval-p a (rcfn-2-domain))
                (inside-interval-p b (rcfn-2-domain))
                (< a b))
           (i-limited (riemann-rcfn-2 (make-small-partition a b) arg))))
\end{Verbatim}

As claimed, the above non-classical theorem can also be applied to
functions with an arbitrary number of free arguments, using the trick
we describe in the following example. In this example, the function
$f(x, m, n)$ contains two free arguments $m$ and $n$. Then, the
parameter {\tt arg} in the above theorem should be considered as the
list {\tt (list m n)}. Having said that, we first need to
prove a lemma stating that {\it every element in a standard list is
standard}. This can be proven easily by using {\tt
defthm-std} \cite{gamboa:2007}.
\begin{Verbatim}
(defthm-std standardp-nth-i-arg
  (implies (and (standardp arg)
                (standardp i))
           (standardp (nth i arg)))
  :rule-classes (:rewrite :type-prescription))
\end{Verbatim}

The functional instantiation with pseudo-lambda expressions can now be
applied to prove the limited property of the Riemann sum of $f$
as follows.
\begin{Verbatim}
(1)  (defthm limited-riemann-f-small-partition-lemma
       (implies (and (standardp arg)
                     (standardp a)
                     (standardp b)
                     (inside-interval-p a (f-domain))
                     (inside-interval-p b (f-domain))
                     (< a b))
                (i-limited (riemann-f (make-small-partition a b)
                                      (nth 0 arg)
                                      (nth 1 arg))))
       :hints (("Goal"
                :by (:functional-instance
                     limited-riemann-rcfn-2-small-partition
                     (rcfn-2 (lambda (x arg)
                               (f x (nth 0 arg) (nth 1 arg))))
                     (rcfn-2-domain f-domain)
                     (map-rcfn-2
                      (lambda (p arg)
                        (map-f p (nth 0 arg) (nth 1 arg))))
                     (riemann-rcfn-2
                      (lambda (p arg)
                        (riemann-f p (nth 0 arg) (nth 1 arg))))))))
\end{Verbatim}

Note that the functional instantiation in the above lemma does not
contain any free arguments. However, this lemma constrains the two
free arguments to be members of a list.  In order to eliminate this
constraint, we need a lemma stating that {\it a list of length two is
standard if both of its elements are standard}.  Again, we can prove
this using {\tt defthm-std}.
\begin{Verbatim}
(defthm-std standardp-list
  (implies (and (standardp m)
                (standardp n))
           (standardp (list m n)))
  :rule-classes (:rewrite :type-prescription))
\end{Verbatim}

We are finally able to prove the desired theorem as an instance of the
lemma (1).
\begin{Verbatim}
(defthm limited-riemann-f-small-partition
  (implies (and (standardp m)
                (standardp n)
                (standardp a)
                (standardp b)
                (inside-interval-p a (f-domain))
                (inside-interval-p b (f-domain))
                (< a b))
           (i-limited (riemann-f (make-small-partition a b) m n)))
  :hints (("Goal"
              :use (:instance limited-riemann-f-small-partition-lemma
                              (arg (list m n))))))
\end{Verbatim}

\subsection{Applying FTC-2 to Functions with Free Arguments}

The FTC-2 theorem was stated and proven in the ACL2 community books for generic
{\it unary} functions as follows \cite{cowles:2014}:
\begin{Verbatim}
(defthm ftc-2
  (implies (and (inside-interval-p a (rcdfn-domain))
                (inside-interval-p b (rcdfn-domain)))
           (equal (int-rcdfn-prime a b)
                  (- (rcdfn b) (rcdfn a)))))
\end{Verbatim}

Again, we would like to apply this theorem for functions with free
arguments via functional instantiation. Since this theorem is
classical, free arguments are allowed to occur in pseudo-lambda
expressions
of a functional substitution as long as classicalness is
preserved \cite{gamboa:2007}. Through functional instantiation with
pseudo-lambda terms, we encounter several proof obligations that require
free arguments to be standard. Unfortunately, attempting to
add this assumption to pseudo-lambda terms, e.g.,
{\tt (lambda (x) (if (standardp n) (f x n) (f x 0)))}, is not allowed
in ACL2(r) since the terms become non-classical by using the
non-classical function {\tt standardp}, violating the classicalness
requirement. To deal with this issue of functional instantiation, we
propose a technique using an encapsulate event with {\it zero-arity
classical functions} (constants) representing free arguments. Since
the zero-arity functions are classical, they must return standard
values. Using this technique, we can instantiate the FTC-2 theorem to
evaluate the definite integral of a function containing free arguments
in terms of its antiderivative. For example, suppose we want to apply
the FTC-2 theorem to a real-valued continuous function $f(x, n)$,
where $n$ is a free argument of type integer. Also suppose that $g$
is an antiderivative of $f$. Our proposed technique consists of four
steps as described below:

\begin{itemize}

\item Step 1: Define an encapsulate event that introduces zero-arity
classical function(s) representing free argument(s).
\begin{Verbatim}
(encapsulate
 (((n) => *))
 (local (defun n () 0))
 (defthm integerp-n
   (integerp (n))
   :rule-classes :type-prescription))
\end{Verbatim}

\item Step 2: Prove that the zero-arity classical function(s) return
   standard values using {\tt defthm-std}.
\begin{Verbatim}
(defthm-std standardp-n
  (standardp (n))
  :rule-classes (:rewrite :type-prescription))
\end{Verbatim}

\item Step 3: Prove the main theorem, modified by replacing
the free argument(s) with the corresponding zero-arity function(s)
introduced in step 1. Without free argument(s), the functional
instantiation can be applied straightforwardly.
\begin{Verbatim}
(defthm f-ftc-2-lemma
  (implies (and (inside-interval-p a (g-domain))
                (inside-interval-p b (g-domain)))
           (equal (int-f a b (n))
                  (- (g b (n))
                     (g a (n)))))
  :hints (("Goal"
            :by (:functional-instance
                  ftc-2
                  (rcdfn
                   (lambda (x) (g x (n))))
                  (rcdfn-prime
                   (lambda (x) (f x (n))))
                  (rcdfn-domain g-domain)
                  ... ;; Instantiate other constrained
                      ;; functions similarly. 
                  (int-rcdfn-prime
                   (lambda (a b) (int-f a b (n))))))))
\end{Verbatim}

\item Step 4: Prove the main theorem by functionally instantiating
the zero-arity function(s) in the lemma introduced in step 3 with the
corresponding free argument(s).
\begin{Verbatim}
(defthm f-ftc-2
  (implies (and (integerp n) ;; we assume the type of n is integer.
                (inside-interval-p a (g-domain))
                (inside-interval-p b (g-domain)))
           (equal (int-f a b n)
                  (- (g b n)
                     (g a n))))
  :hints (("Goal"
           :by (:functional-instance f-ftc-2-lemma
                                     (n (lambda ()
                                          (if (integerp n) n 0)))))))
\end{Verbatim}

\end{itemize}

\section{Orthogonality Relations of Trigonometric Functions}
\label{sec:orthog-relations}
By applying the FTC-2 evaluation procedure, we can mechanically prove
in ACL2(r) the orthogonality relations of trigonometric functions,
which are the essential properties in Fourier series analysis. The
orthogonality relations of trigonometric functions are a collection of
definite integral formulas for sine and cosine functions as described
below:
\numberwithin{equation}{section}
\begin{equation} \label{sines-orthog}
\int_{-L}^L \sin(m \frac{\pi}{L} x) \sin(n \frac{\pi}{L} x)dx =
\begin{cases}
0, \mbox{ if } m \neq n \lor m = n = 0 \\
L, \mbox{ if } m = n \neq 0
\end{cases}
\end{equation}
\begin{equation} \label{cosines-orthog}
\int_{-L}^L \cos(m \frac{\pi}{L} x) \cos(n \frac{\pi}{L} x)dx =
\begin{cases}
0, \mbox{ if } m \neq n \\
L, \mbox{ if } m = n \neq 0 \\
2L, \mbox{ if } m = n = 0
\end{cases}
\end{equation}
\begin{equation} \label{sine-cosine-orthog}
\int_{-L}^L \sin(m \frac{\pi}{L} x) \cos(n \frac{\pi}{L} x)dx = 0
\end{equation}

\noindent where $x, L \in \mathbb{R}$; $L \neq 0$; and $m, n \in \mathbb{N}$. As
mentioned, these integral formulas can be proven using the FTC-2
evaluation procedure. Let's consider the case $m \neq n$ in formula
(\ref{sines-orthog}); the other cases can be proven similarly. When $m \neq
n$, formula (\ref{sines-orthog}) states that $\int_{-L}^L f(x, m, n,
L) dx = 0$ where $f(x, m, n, L) = \sin(m \frac{\pi}{L}
x) \sin(n \frac{\pi}{L} x)$. Using the automatic differentiator, we
can easily prove that the function $g$ defined below is indeed an
antiderivative of $f$ when $m \neq n$:
\[
g(x, m, n, L) = \frac1{2} \left(\frac{\sin \left((m - n) \frac{\pi}{L}
x \right)}{(m - n) \frac{\pi}{L}} - \frac{\sin \left((m +
n) \frac{\pi}{L} x \right)}{(m + n) \frac{\pi}{L}} \right)
\]

\noindent Then, by the FTC-2 theorem,
\begin{flalign*}
\int_{-L}^L f(x, m, n, L) dx &= g(L, m, n, L) - g(-L, m, n, L) \\
&= \frac1{2} \left(\frac{\sin \left((m - n) \pi \right)}{(m -
n) \frac{\pi}{L}} - \frac{\sin \left((m + n) \pi \right)}{(m +
n) \frac{\pi}{L}} \right) - \frac1{2} \left(\frac{\sin \left(- (m -
n) \pi \right)}{(m - n) \frac{\pi}{L}} - \frac{\sin \left(- (m +
n) \pi \right)}{(m + n) \frac{\pi}{L}} \right) \\
&= \frac1{2}(0 - 0) - \frac1{2}(0 - 0) = 0
\end{flalign*}

\section{Sum Rule for Definite Integrals of Indexed Sums}
\label{sec:sum-rule-for-integration}
As part of the Fourier coefficient formalization, we need to formalize
the sum rule for definite integrals of indexed sums, which is stated
as the following theorem:

\begin{theorem}[Sum rule for definite integrals of indexed sums]
Let \{$f_n$\} be a set of real-valued continuous functions on $[a, b]$,
where $n = 0, 1, 2,..., N$. Then
\[
\int_{a}^b \sum_{n=0}^N f_n(x)dx = \sum_{n=0}^N \int_{a}^b f_n(x)dx
\]

\end{theorem}

Note that $f_n(x)$ abbreviates
$f(x, n)$, which contains a free argument, $n$. Kaufmann
\cite{kaufmann:2000} formalized the FTC-1 theorem for generic unary
functions as a non-classical theorem. We re-prove it for generic
functions with an extra argument added, following the method for
extending the limited property of Riemann sums as described in Section
\ref{subsec:riemann-int}. Then, by applying the FTC-2 evaluation
procedure along with the extended version of FTC-1 and the sum rule
for differentiation, the above theorem can be proven as follows:

\begin{proof}
For all $x \in [a, b]$ and $n = 0, 1,\ldots,N$, let
\[
g_n(x) = \int_a^x f_n(t) dt.
\]

By FTC-1, $g_n'(x) = f_n(x)$ for all $x \in [a, b], n = 0, 1,...,N$.

By the sum rule for differentiation, $\left(\sum_{n=0}^N g_n(x)
\right)' = \sum_{n=0}^N g_n'(x) = \sum_{n=0}^N f_n(x)$ for all $x \in
      [a, b]$. Then, by FTC-2,
\begin{flalign*}
\int_{a}^b \sum_{n=0}^N f_n(x)dx &= \sum_{n=0}^N g_n(b) - \sum_{n=0}^N
g_n(a) & \\
&= \sum_{n=0}^N \int_a^b f_n(t) dt - \sum_{n=0}^N \int_a^a
f_n(t) dt = \sum_{n=0}^N \int_{a}^b f_n(x)dx
\end{flalign*}
\end{proof}

\section{Fourier Coefficient Formulas}
\label{sec:fourier-coefficients}
\numberwithin{equation}{section}

From the orthogonality relations and the sum rule for integration, the
Fourier coefficients of periodic functions can be stated as follows:

\begin{theorem}[Fourier coefficient formulas]
Consider the following Fourier sum $f(x)$ for a periodic function with
period $2L$:
\begin{equation} \label{fourier-sum}
f(x) = a_0 + \sum_{n=1}^N \left(a_n \cos (n \frac{\pi}{L} x) + b_n
\sin (n \frac{\pi}{L} x) \right)
\end{equation}

\noindent Then
\begin{equation} \label{a_0}
a_0 = \frac{1}{2L} \int_{-L}^L f(x)dx,
\end{equation}
\begin{equation} \label{a_n}
a_n = \frac{1}{L} \int_{-L}^L f(x) \cos (n \frac{\pi}{L} x)dx,
\end{equation}
\begin{equation} \label{b_n}
b_n = \frac{1}{L} \int_{-L}^L f(x) \sin (n \frac{\pi}{L} x)dx.
\end{equation}

\end{theorem}

The proof of this theorem is straightforward from the orthogonality
relations and the sum rule for integration, after applying the
definition of $f$ in (\ref{fourier-sum}) to the Fourier coefficient
formulas (\ref{a_0}), (\ref{a_n}), and (\ref{b_n}). Consequently, we
can easily derive the following corollary. This is known as the
uniqueness of Fourier sums:

\begin{corollary}[Uniqueness of Fourier sums]
Let
\[
f(x) = a_0 + \sum_{n=1}^N \left(a_n \cos (n \frac{\pi}{L} x) + b_n
\sin (n \frac{\pi}{L} x) \right)
\]

\noindent and
\[
g(x) = A_0 + \sum_{n=1}^N \left(A_n \cos (n \frac{\pi}{L} x) + B_n
\sin (n \frac{\pi}{L} x) \right)
\]

\noindent Then $f = g \Leftrightarrow
\begin{cases}
a_0 = A_0 \\
a_n = A_n, \mbox{ for all } n = 1, 2,..., N \\
b_n = B_n, \mbox{ for all } n = 1, 2,..., N
\end{cases}
$

\end{corollary}

\begin{proof}
($\Rightarrow$) Follows immediately from the Fourier coefficient
  formulas:
\[
a_0 = \frac{1}{2L} \int_{-L}^L f(x) dx
= \frac{1}{2L} \int_{-L}^L g(x) dx = A_0,
\]
\[
a_n = \frac{1}{L} \int_{-L}^L f(x) \cos (n \frac{\pi}{L} x) dx
= \frac{1}{L} \int_{-L}^L g(x) \cos (n \frac{\pi}{L} x) dx = A_n,
\]
\[
b_n = \frac{1}{L} \int_{-L}^L f(x) \sin (n \frac{\pi}{L} x) dx
= \frac{1}{L} \int_{-L}^L g(x) \sin (n \frac{\pi}{L} x) dx = B_n.
\]

($\Leftarrow$) Obviously true by induction on $n$.
\end{proof}

\section{Sum Rule for Definite Integrals of Infinite Series}
\label{sec:int-inf}
\numberwithin{equation}{section}

Our formalization of the sum rule for integration in
Section \ref{sec:sum-rule-for-integration} only applies to finite
sums. However, Fourier series can be infinite. Thus, the sum rule for
integration needs to be extended to infinite series; we do so in this section.

Our basic result is the following sum rule for integrals of infinite
series, formalized using non-standard analysis.  (We define uniform
convergence below.)

\begin{theorem}\label{int-inf}  Suppose that $\{f_n(x)\}$ is a
sequence of real-valued continuous functions whose sequence of partial
sums converges uniformly to a continuous limit function on a given
interval. Then
\begin{equation}
\int_a^b \mbox{st} \left(\sum_{n=0}^{H_0} f_n(x) \right) dx =
\mbox{st} \left(\sum_{n=0}^{H_1} \int_a^b f_n(x) dx \right)
\end{equation}

\noindent for all {\it infinitely large} natural numbers $H_0$ and $H_1$.
\end{theorem}

{\bf Remark}. The conclusion above is equivalent to the
following formula using the {\it epsilon-delta} definition of
limit \cite{keisler:1985}:
\begin{equation} \label{int-inf-alt}
\int_a^b \lim_{N \to \infty} \left(\sum_{n=0}^N f_n(x) \right) dx =
\lim_{N \to \infty} \left(\sum_{n=0}^N \int_a^b f_n(x) dx \right)
\end{equation}

We turn now to two variants of this theorem that relax its
requirements.  We start by recalling well-known formulations of
convergence in non-standard analysis.

\begin{itemize}

\item {\bf Pointwise convergence:} Suppose \{$f_n$\} is a sequence of
  functions defined on an interval $I$.  The sequence \{$f_n$\}
  converges {\it pointwise} to the {\it limit function} $f$ on the
  interval $I$ if $f_H(x) \approx f(x)$ for all {\it standard} $x \in
  I$ and for all infinitely large natural numbers $H$.

\item {\bf Uniform convergence:} Suppose \{$f_n$\} is a sequence of
  functions defined on an interval $I$.  The sequence \{$f_n$\}
  converges {\it uniformly} to the {\it limit function} $f$ on the
  interval $I$ if $f_H(x) \approx f(x)$ for all $x \in I$ (both {\it
    standard} and {\it non-standard}) and for all infinitely large
  natural numbers $H$.

\end{itemize}

Clearly, uniform convergence is stronger than pointwise convergence. A
sequence that converges uniformly to a limit function also converges
pointwise to that function, but the reverse is not guaranteed. We meet
the hypothesis of Theorem~\ref{int-inf} --- uniform convergence to a
continuous limit function --- in two ways corresponding to two
different conditions, as follows.  Note: Only the second condition is
relevant to Fourier series, but the first also leads to an interesting
result.

\begin{itemize}

\item {\bf Condition 1:} A {\it monotone} sequence of partial sums of
  real-valued continuous functions {\it converges pointwise} to a {\it
  continuous} limit function on the {\it closed and bounded} interval
  of interest.

\item {\bf Condition 2:} A sequence of partial sums of real-valued
  continuous functions {\it converges uniformly} to a limit function
  on the interval of interest.

\end{itemize}

The following theorem~\cite{luxemburg:1971} shows that Condition 1
implies the hypothesis of Theorem~\ref{int-inf}.

\begin{theorem}[Dini Uniform Convergence Theorem]
\label{thm:dini}
A monotone sequence of continuous functions \{$f_n$\} that converges
pointwise to a continuous function $f$ on a closed and bounded
interval $[a, b]$ is uniformly convergent.
\end{theorem}

Our proof of Dini's theorem relies on the {\it
overspill principle} from non-standard analysis \cite{goldblatt:1998,
keisler:1976}. Thus, we now discuss our formalization of this
principle in ACL2(r)~\cite{overspill}.

{\bf Overspill principle (weak version):} Let $P(n, x)$ be a classical
predicate. Then
\begin{equation} \label{overspill-1}
\forall x. \left( \left( \forall^{st} n \in \mathbb{N}. P(n, x)
\right) \Rightarrow \exists^{\neg st} k \in \mathbb{N}. P(k, x)
\right).
\end{equation}

\noindent In words, if a classical predicate $P$ holds for all standard natural
numbers $n$, $P$ must be hold for some non-standard natural number
$k$. By applying this principle, we can even come up with a stronger
statement as follows:

{\bf Overspill principle (strong version):} Let $P(n, x)$ be a
classical predicate. Then
\begin{equation} \label{overspill-2}
\forall x. \left( \left( \forall^{st} n \in \mathbb{N}. P(n, x)
\right) \Rightarrow \exists^{\neg st} k \in \mathbb{N}. \forall m
\in \mathbb{N}. \left( m \leq k \Rightarrow P(m, x) \right) \right).
\end{equation}

\noindent In words, if a classical predicate $P$ holds for all standard natural
numbers $n$, there must exist some non-standard natural number $k$
such that $P$ holds for all natural numbers less than or equal to
$k$. (\ref{overspill-2}) can be derived from (\ref{overspill-1})
through a classical predicate $P^*(n, x)$ defined in terms of $P(n,
x)$ as follows:
\begin{Verbatim}
(defun P* (n x)
  (if (zp n)
      (P 0 x)
    (and (P n x)
         (P* (1- n) x))))
\end{Verbatim}

\noindent The proof of (\ref{overspill-2}) now proceeds as follows.

\begin{enumerate}

\item Fix $x$ and assume $\left( \forall^{st} n \in \mathbb{N}. P(n, x) \right)$.

\item Then, we can show that $\left( \forall^{st} m \in
\mathbb{N}. P^*(m, x) \right)$.

\item Applying (\ref{overspill-1}) to $P^*$:
$\exists^{\neg st} k \in \mathbb{N}. P^*(k, x)$.

\item From the definition of $P^*$: $\exists^{\neg st} k \in
\mathbb{N}. \forall m \in
\mathbb{N}. \left( m \leq k \Rightarrow P(m, x) \right)$.

\end{enumerate}

We formalize (\ref{overspill-2}) for a generic classical predicate
$P(n, x)$ in ACL2(r) and provide the {\tt overspill} utility, which
automates the application of (\ref{overspill-2}). In particular, the
user needs only to define a classical predicate $P_1$ and then call
the {\tt overspill} macro with the input $P_1$ so that
(\ref{overspill-2}) will be applied to $P_1$ automatically via a
functional instantiation. We can thus apply the overspill
principle (\ref{overspill-2}) to prove Dini's theorem as shown
below.

\medskip
\noindent
{\it Proof of Theorem \ref{thm:dini}.} Without loss of generality,
assume \{$f_n$\} is {\it monotonically increasing}. We want to prove
$f(x) \approx f_H(x)$ for all $x \in [a, b]$ and for all infinitely
large $H \in \mathbb{N}$.

Fact: If $x \in [a, b]$ then st$(x) \in [a, b]$ (note that this is
only true on {\it closed and bounded intervals}).

(A) Since st$(x)$ is standard and $x \approx \mbox{st} (x)$, $f(x) \approx
f(\mbox{st} (x))$ by the continuity of $f$.

(B) Since st$(x)$ is standard, $f(\mbox{st} (x)) \approx f_H(\mbox{st}
(x))$ by the pointwise convergence of \{$f_n$\}.

We will make the following two claims.

Claim 1. For some non-standard $k \in \mathbb{N}$, we have: for all
$H \leq k$, $f_H(\mbox{st} (x)) \approx f_H(x)$.

Claim 2. Suppose $k \in \mathbb{N}$ such that $f_k(x) \approx
f(x)$.  Then for all $H > k$, $f(x) \approx f_H(x)$.

For the moment, assume both claims.  Choosing $k$ according to Claim
1, then by the transitivity of {\tt i-close} and Steps (A) and (B)
above, we have $f(x) \approx f_H(x)$ for all infinitely large $H \leq
k$.  Applying Claim 2 to that same $k$ takes care of $H > k$, so we
are done once we prove the claims.

To prove Claim 1, we must find a non-standard $k$ such that
$f_H(\mbox{st} (x)) \approx f_H(x)$ for all $H \leq k$. We first
observe that by the continuity of \{$f_n$\}, we have $f_n(\mbox{st}
(x)) \approx f_n(x), \forall x \in [a, b]$ and $\forall^{st}
n \in \mathbb{N}$.
We apply the overspill principle --- which requires a classical
predicate --- to the following classical predicate $P(n, x_0, x)$.
\[
P(n, x_0, x) \equiv |f_n(x_0) - f_n(x)| < \frac{1}{n + 1}
\]

\noindent If $x_0, x \in [a, b]$, $x_0$ is standard, and $x_0 \approx x$, then
$P(n, x_0, x)$ holds for all {\it standard} $n \in \mathbb{N}$ since
$f_n(x_0) - f_n(x) \approx 0$ by the continuity of \{$f_n$\}. Hence,
by the overspill principle (\ref{overspill-2}), there exists a {\it
non-standard} $k \in \mathbb{N}$ s.t. $P(m, x_0, x)$ holds for all
$m \in \mathbb{N}$ and $m \leq k$.
Now suppose that $H$ is infinitely large but $H \leq k$.  Then
$f_H(x_0) \approx f_H(x)$ since

\[
0 \leq |f_H(x_0) - f_H(x)| < \frac{1}{H + 1} \approx 0.
\]

\noindent Let's pick $x_0$ to be st($x$); then $f_H(\mbox{st} (x)) \approx
f_H(x)$, concluding the proof of Claim 1.

To prove Claim 2, by hypothesis pick $k \in \mathbb{N}$ such that
$f_k(x) \approx f(x)$, and assume $H > k$. Then $f_k(x) \leq
f_H(x) \leq f(x)$ by the increasing monotonicity of \{$f_n$\}.
Hence, $0 \leq |f(x) - f_H(x)| \leq |f(x) - f_k(x)| \approx 0$.
Thus, $f(x) \approx f_H(x)$, which concludes the proof of Claim 2 and
also the proof of Theorem~\ref{thm:dini}. \hfill $\square$

\medskip

Dini's theorem shows that pointwise convergence of a sequence of
continuous functions on a closed and bounded interval also implies its
uniform convergence if the sequence is monotonic and the limit
function is continuous. Unfortunately, it is not applicable to Fourier
series since Fourier series are not required to be monotonic. As a
result, Fourier series cannot meet the requirement for our proof of
(\ref{int-inf}) from Condition 1. In fact, Fourier series can satisfy
Condition 2 under suitable criteria \cite{jackson:1934}. Then, from
Condition 2, we need to prove that the limit function is continuous in
order to meet the requirement for our proof of (\ref{int-inf}). This
can be proven by applying the overspill principle.

\medskip

\begin{theorem}
\label{thm:continuity-limit-function}
Suppose that a sequence of continuous functions \{$f_n$\} converges
uniformly to a limit function $f$ on an interval $I$. Then $f$ is also
continuous on $I$.
\end{theorem}

\begin{proof}
The goal is to prove $f(x_0) \approx f(x)$ for all $x_0, x \in
I$ such that $x_0$ is standard and $x_0 \approx x$. By the uniform convergence
of \{$f_n$\}, we have $f(x_0) \approx f_H(x_0)$ and $f(x) \approx
f_H(x)$ for all infinitely large $H \in \mathbb{N}$. If we can show
that $f_H(x_0) \approx f_H(x)$, then we obtain our goal by the
transitivity of $\approx$. By applying the overspill principle in the
same way as in our proof of Theorem \ref{thm:dini}, we claim that
there must exist a non-standard $k \in \mathbb{N}$
s.t. $f_H(x_0) \approx f_H(x)$ if $H \leq k$. When $H > k$, we know
that $f_H(x_0) \approx f_k(x_0)$ (since they are both {\tt i-close}
to $f(x_0)$ by the uniform convergence of \{$f_n$\}) and similarly
$f_H(x) \approx f_k(x)$. Thus, $f_H(x_0) \approx f_k(x_0) \approx
f_k(x) \approx f_H(x)$ and we are done.
\end{proof}

From Condition 2 and Theorem \ref{thm:continuity-limit-function}, in
order to apply the sum rule for integration (\ref{int-inf}) to
infinite Fourier series, we need to prove that the Fourier series
converge uniformly to limit functions. As mentioned above, this is
provable under suitable criteria \cite{jackson:1934}.

\section{Conclusions}
\label{sec:conclusions}
We described in this paper our extension of the framework for formally
evaluating definite integrals of real-valued continuous functions
containing free arguments, using FTC-2. Along with this extension, we
also presented our technique for handling the occurrence of free
arguments in pseudo-lambda expressions of functional
instantiations. Using the extended framework, we showed how to prove
the orthogonality relations of trigonometric functions as well as the
sum rule for definite integrals of indexed sums. These properties were
then applied to prove the Fourier coefficient formulas and
consequently used to derive the uniqueness of Fourier sums as a corollary.

We also presented our formalization of the sum rule for definite
integrals of infinite series under two different conditions. Along
with this task, we formalized the overspill principle and provided the
{\tt overspill} utility that automates the application of the overspill
principle, thus strengthening the reasoning capability of non-standard
analysis in ACL2(r). Our proofs of Dini's theorem and the continuity
of the limit function as described in Section \ref{sec:int-inf}
illustrate this capability.

Some possible areas of future work are worth mentioning. First, the
automatic differentiator needs to be extended to support partial
differentiation. The current AD has limited support for automating partial
differentiation.
Although we extended the AD to support partial derivative
registrations of binary functions, this extension is still very
limited for automatic differentiation. In particular, our extension
imposes a constraint on the free argument of binary functions that
either its symbolic name must be $arg0$ or it has to be a constant. As
a result, the AD cannot be applied to expressions containing several
binary functions with different free arguments. Our current solution
in this case is to break those expressions into smaller expressions
such that the AD can be applied directly to these smaller expressions,
and then manually combine them to get the final result for the
original expressions. Future work might make the partial
differentiation process more automatic.  Another possibility for
future work is to prove convergence of the Fourier series for a
periodic function, under sufficient conditions.

In summary, we have developed and extended frameworks for mechanized
continuous mathematics, which we applied to obtain results about
Fourier series and an elegant proof of Dini's theorem.  We are
confident that our frameworks can be applied to future work on Fourier
series and, more generally, continuous mathematics, to be carried out
in ACL2(r).

\section*{Acknowledgements}

We thank Ruben Gamboa for useful discussions.  We also thank the
reviewers for useful comments.  This material is based
upon work supported by DARPA under Contract No. N66001-10-2-4087.

\nocite{*}
\bibliographystyle{eptcs}
\bibliography{generic}

\begin{thebibliography}{10}
\providecommand{\bibitemdeclare}[2]{}
\providecommand{\surnamestart}{}
\providecommand{\surnameend}{}
\providecommand{\urlprefix}{Available at }
\providecommand{\url}[1]{\texttt{#1}}
\providecommand{\href}[2]{\texttt{#2}}
\providecommand{\urlalt}[2]{\href{#1}{#2}}
\providecommand{\doi}[1]{doi:\urlalt{http://dx.doi.org/#1}{#1}}
\providecommand{\bibinfo}[2]{#2}

\bibitemdeclare{misc}{acl2-doc:lemma-instance}
\bibitem{acl2-doc:lemma-instance}
\bibinfo{author}{\surnamestart {ACL2}\surnameend}: \emph{\bibinfo{title}{{ACL2
  Documentation on Lemma-Instance}}}.
\newblock \bibinfo{note}{See URL
  \url{http://www.cs.utexas.edu/users/moore/acl2/current/manual/index.html?topic=ACL2____LEMMA-INSTANCE}}.

\bibitemdeclare{misc}{acl2-books}
\bibitem{acl2-books}
\bibinfo{author}{ACL2~Community \surnamestart Books\surnameend}:
  \urlprefix\url{https://github.com/acl2/acl2}.

\bibitemdeclare{misc}{overspill}
\bibitem{overspill}
\bibinfo{author}{Overspill Principle Formalization~Source \surnamestart
  Code\surnameend}:
  \urlprefix\url{https://raw.githubusercontent.com/acl2/acl2/master/books/nonstd/nsa/overspill.lisp}.

\bibitemdeclare{article}{gamboa:2007}
\bibitem{gamboa:2007}
\bibinfo{author}{R.~Gamboa \&~J. \surnamestart Cowles\surnameend}
  (\bibinfo{year}{2007}): \emph{\bibinfo{title}{Theory Extension in ACL2(r)}}.
\newblock {\sl \bibinfo{journal}{Journal of Automated Reasoning}}
  \bibinfo{volume}{38}(\bibinfo{number}{4}), pp. \bibinfo{pages}{273--301},
  \doi{10.1007/s10817-006-9043-0}.

\bibitemdeclare{inproceedings}{cowles:2014}
\bibitem{cowles:2014}
\bibinfo{author}{J.~Cowles \&~R. \surnamestart Gamboa\surnameend}
  (\bibinfo{year}{2014}): \emph{\bibinfo{title}{Equivalence of the Traditional
  and Non-Standard Definitions of Concepts from Real Analysis}}.
\newblock In: {\sl \bibinfo{booktitle}{Proc of the Twelfth International
  Workshop on the ACL2 Theorem Prover and its Applications (ACL2-2014)}}, pp.
  \bibinfo{pages}{89--100}, \doi{10.4204/EPTCS.152.8}.

\bibitemdeclare{inproceedings}{reid:itp-2011}
\bibitem{reid:itp-2011}
\bibinfo{author}{P.~Reid \&~R. \surnamestart Gamboa\surnameend}
  (\bibinfo{year}{2011}): \emph{\bibinfo{title}{Automatic Differentiation in
  ACL2}}.
\newblock In: {\sl \bibinfo{booktitle}{Proc of the Second International
  Conference on Interactive Theorem Proving (ITP-2011)}}, pp.
  \bibinfo{pages}{312--324}, \doi{10.1007/978-3-642-22863-6\_23}.

\bibitemdeclare{inproceedings}{reid:acl2-2011}
\bibitem{reid:acl2-2011}
\bibinfo{author}{P.~Reid \&~R. \surnamestart Gamboa\surnameend}
  (\bibinfo{year}{2011}): \emph{\bibinfo{title}{Implementing an Automatic
  Differentiator in ACL2}}.
\newblock In: {\sl \bibinfo{booktitle}{Proc of the Tenth International Workshop
  on the ACL2 Theorem Prover and its Applications (ACL2-2011)}}, pp.
  \bibinfo{pages}{61--69}, \doi{10.4204/EPTCS.70.5}.

\bibitemdeclare{phdthesis}{gamboa:1999}
\bibitem{gamboa:1999}
\bibinfo{author}{R.~\surnamestart Gamboa\surnameend} (\bibinfo{year}{1999}):
  \emph{\bibinfo{title}{Mechanically Verifying Real-Valued Algorithms in
  ACL2}}.
\newblock Ph.D. thesis, \bibinfo{school}{The University of Texas at Austin}.

\bibitemdeclare{book}{goldblatt:1998}
\bibitem{goldblatt:1998}
\bibinfo{author}{R.~\surnamestart Goldblatt\surnameend} (\bibinfo{year}{1998}):
  \emph{\bibinfo{title}{Lectures on the Hyperreals: An Introduction to
  Nonstandard Analysis}}.
\newblock \bibinfo{publisher}{Springer}.

\bibitemdeclare{article}{jackson:1934}
\bibitem{jackson:1934}
\bibinfo{author}{D.~\surnamestart Jackson\surnameend} (\bibinfo{year}{1934}):
  \emph{\bibinfo{title}{The Convergence of Fourier Series}}.
\newblock {\sl \bibinfo{journal}{The American Mathematical Monthly}}
  \bibinfo{volume}{41}(\bibinfo{number}{2}), pp. \bibinfo{pages}{67--84},
  \doi{10.2307/2300327}.

\bibitemdeclare{incollection}{kaufmann:2000}
\bibitem{kaufmann:2000}
\bibinfo{author}{M.~\surnamestart Kaufmann\surnameend} (\bibinfo{year}{2000}):
  \emph{\bibinfo{title}{Modular Proof: The Fundamental Theorem of Calculus}}.
\newblock In \bibinfo{editor}{M.~\surnamestart Kaufmann\surnameend},
  \bibinfo{editor}{P.~\surnamestart Manolios\surnameend} \&
  \bibinfo{editor}{J~S. \surnamestart Moore\surnameend}, editors: {\sl
  \bibinfo{booktitle}{Computer-Aided Reasoning: ACL2 Case Studies}},
  chapter~\bibinfo{chapter}{6}, \bibinfo{volume}{4},
  \bibinfo{publisher}{Springer US}, pp. \bibinfo{pages}{75--91},
  \doi{10.1007/978-1-4757-3188-0\_6}.

\bibitemdeclare{article}{gamboa:2001}
\bibitem{gamboa:2001}
\bibinfo{author}{R.~Gamboa \&~M. \surnamestart Kaufmann\surnameend}
  (\bibinfo{year}{2001}): \emph{\bibinfo{title}{Non-Standard Analysis in
  ACL2}}.
\newblock {\sl \bibinfo{journal}{Journal of Automated Reasoning}}
  \bibinfo{volume}{27}(\bibinfo{number}{4}), pp. \bibinfo{pages}{323--351},
  \doi{10.1023/A:1011908113514}.

\bibitemdeclare{book}{keisler:1976}
\bibitem{keisler:1976}
\bibinfo{author}{H.~J. \surnamestart Keisler\surnameend}
  (\bibinfo{year}{1976}): \emph{\bibinfo{title}{Foundations of Infinitesimal
  Calculus}}.
\newblock \bibinfo{publisher}{Prindle Weber \& Schmidt}.

\bibitemdeclare{book}{keisler:1985}
\bibitem{keisler:1985}
\bibinfo{author}{H.~J. \surnamestart Keisler\surnameend}
  (\bibinfo{year}{1985}): \emph{\bibinfo{title}{Elementary Calculus: An
  Infinitesimal Approach}}.
\newblock \bibinfo{publisher}{Prindle Weber \& Schmidt}.

\bibitemdeclare{article}{luxemburg:1971}
\bibitem{luxemburg:1971}
\bibinfo{author}{W.~A.~J. \surnamestart Luxemburg\surnameend}
  (\bibinfo{year}{1971}): \emph{\bibinfo{title}{Arzela's Dominated Convergence
  Theorem for the Riemann Integral}}.
\newblock {\sl \bibinfo{journal}{The American Mathematical Monthly}}
  \bibinfo{volume}{78}(\bibinfo{number}{9}), pp. \bibinfo{pages}{970--979},
  \doi{10.2307/2317801}.

\bibitemdeclare{article}{mathematica:2015}
\bibitem{mathematica:2015}
\bibinfo{author}{Inc. \surnamestart Wolfram~Research\surnameend}
  (\bibinfo{year}{2015}): \emph{\bibinfo{title}{Mathematica}}.
\newblock \urlprefix\url{http://www.wolfram.com/mathematica/}.

\end{thebibliography}
\end{document}